\documentclass[11pt]{article}

\usepackage{hyperref}
\hypersetup{colorlinks=true,citecolor=blue,linkcolor=blue,urlcolor=blue}
\usepackage{graphicx} % Required for inserting images
\usepackage{amsthm}
\usepackage{amsmath}
\usepackage{amssymb}  % in your preamble
\usepackage{amsfonts}
\usepackage{bbm}

\newtheorem{theorem}{Theorem}[section]
\newtheorem*{theorem*}{Theorem}
\newtheorem{corollary}[theorem]{Corollary}
\newtheorem{lemma}[theorem]{Lemma}

\theoremstyle{definition}

\newtheorem*{remark*}{Remark}
\usepackage{enumitem}

\theoremstyle{remark}
\newtheorem*{remark}{Remark}
\usepackage{color}
\usepackage[margin=1in]{geometry}

\newcommand{\E}{\mathbb{E}}

\def\id{\text{id}}
\def\TRUE{\mathsf{TRUE}}
\def\FALSE{\mathsf{FALSE}}
\def\notalltrue{\Omega_n^-}
\def\minfalse{\mathrm{MinFalse}}
\def\smallsetfalse{S}
\def\extminfalse{\mathrm{ExtMinFalse}}

\title{One-Shot Learning for $k$-SAT}
\author{
 Andreas Galanis, Leslie Ann Goldberg, Xusheng Zhang \\
  \small University of Oxford, Oxford OX1 3QD, UK
}

\begin{document}
\maketitle
\let\thefootnote\relax\footnotetext{For the purpose of Open Access, the authors have applied a CC BY public copyright licence to any Author Accepted Manuscript version arising from this submission. All data is provided in full in the results section of this paper.}
\begin{abstract}
%% Text of abstract
Consider  a $k$-SAT formula $\Phi$ where every variable appears at most $d$ times, and let $\sigma$ be a satisfying assignment of $\Phi$ sampled proportionally to $e^{\beta m(\sigma)}$ where $m(\sigma)$ is the number of variables set to true and $\beta$ is a real parameter. Given $\Phi$ and $\sigma$, can we learn the value of $\beta$ efficiently?

This problem falls into a recent line of works about single-sample (``one-shot'') learning of Markov random fields. The $k$-SAT setting we consider here was recently studied by Galanis, Kalavasis, Kandiros (SODA 2024).   They showed that single-sample learning is possible when roughly $d\leq 2^{k/6.45}$ and  impossible when $d\geq (k+1) 2^{k-1}$. Crucially, for their impossibility results they used the existence of unsatisfiable instances which, aside from the gap in $d$, left open the question of whether the feasibility threshold for one-shot learning is dictated by the satisfiability threshold of $k$-SAT formulas of bounded degree.

Our main contribution is to answer this question negatively. We show that one-shot learning for $k$-SAT is infeasible well below the satisfiability threshold; in fact, we obtain impossibility results for degrees $d$ as low as $k^2$ when $\beta$ is sufficiently large, and bootstrap this to small values of $\beta$ when $d$ scales exponentially with $k$, via a probabilistic construction.  On the positive side, we simplify the analysis of the learning algorithm and obtain significantly stronger bounds on $d$ in terms of $\beta$. In particular, for the uniform case $\beta\rightarrow 0$ that has been studied extensively in the sampling literature, our analysis shows that learning is possible under the condition $d\lesssim 2^{k/2}$. 
Note that this is (up to constant factors) all the way to the sampling threshold -- it is   known that  sampling a uniformly-distributed satisfying assignment is NP-hard for $d\gtrsim 2^{k/2}$. 
\end{abstract}

\section{Introduction}
A key task that arises in statistical inference is to estimate the underlying parameters of a distribution, frequently based on the assumption that one has access to a sufficiently large number of independent and identically distributed (i.i.d.) samples. However, in many settings it is critical to perform the estimation with substantially fewer samples,  driven by constraints in data availability, computational cost, or real-time decision-making requirements. In this paper, we consider the extreme setting where only a single sample is available and investigate the feasibility of parameter estimation in this case. We refer to this setting as ``one-shot learning''.

Markov random fields (also known as undirected graphical models) are a canonical framework used to model high-dimensional distributions. The seminal work of \cite{Chatterjee07c} initiated the study of one-shot learning for the Ising and spin glass models, a significant class of Markov random fields that includes the well-known Sherrington-Kirkpatrick and Hopfield models. This approach was later explored in greater depth for the 
Ising model by \cite{BM18} and subsequently extended to tensor or weighted variants of the Ising model in \cite{GM20, MSB, DDDVK21}.
 Beyond the Ising model,  \cite{DDK19, DDP20}  examined one-shot learning in more general settings, notably including logistic regression and higher-order spin systems, obtaining  various algorithmic results in ``soft-constrained''  models, i.e., models where the distribution is supported on the entire state space. \cite{BR21} showed that efficient parameter estimation using one sample is still possible under the presence of hard constraints which prohibit certain states,  relaxing the soft-constrained assumption with ``permissiveness''; canonical Markov random fields in this class include various combinatorial models such as the hardcore model (weighted independent sets). Notably, in all these cases, one-shot learning is always feasible with mild average-degree assumptions on the underlying graph (assuming of course access to an appropriate sample).

More recently, \cite{GKK24} investigated one-shot learning for hard-constrained models that are not permissive, focusing primarily on $k$-SAT and proper colourings; in contrast to soft-constrained models, they showed that one-shot learning is not always possible and investigated its feasibility under various conditions. Their results left however one important question open for $k$-SAT, in terms of identifying the ``right'' feasibility threshold. In particular, their impossibility results were based on the existence of unsatisfiable instances for $k$-SAT, suggesting  that  it might be the satisfiability threshold that is most relevant for one-shot learning. Here we refute this in a strong way. We show infeasibility well below the satisfiability threshold, and obtain positive results that align closely with the conjectured threshold for sampling satisfying assignments.

\subsection{Definitions and Main Results}

In the $k$-SAT model, we consider the state space $\Omega_n:=\{\TRUE, \FALSE\}^n$, where each element is an assignment to $n$ Boolean variables.
The support of the Markov random field is then restricted to the set of assignments that satisfy a given $k$-CNF formula.
More precisely, we define $\Phi_{n,k,d}$ as the set of CNF formulas with $n$ variables such that each clause has exactly $k$ distinct variables and each variable appears in at most $d$ clauses. For an assignment $\sigma\in \Omega_n$ and a formula $\Psi\in \Phi_{n,k,d}$, we denote by $\sigma \models \Psi$ the event that $\sigma$ satisfies $\Psi$ and we denote by $m(\sigma)$
the number of variables that are assigned to~$\TRUE$ in~$\sigma$. (See  Section~\ref{sec:prelim} for further details.)

We study the weighted $k$-SAT model parametrized by $\beta$. 
For a fixed formula $\Psi\in \Phi_{n,k,d}$, the probability for each assignment $\sigma\in \Omega_n$ is given by
\begin{equation}\label{eq:model}
    \Pr_{\Psi, \beta}[\sigma] = 
\frac{ e^{\beta \, m(\sigma)} \,\mathbbm{1}[\sigma \models \Psi]}{
\sum_{\sigma \in \Omega_n} \, e^{\beta \, m(\sigma)} \,\mathbbm{1}[\sigma \models \Psi]
},
\end{equation}
Let $\Omega_n(\Psi):=\{\sigma\in \Omega_n: \sigma\models \Psi\}$ be the support of $\Pr_{\Psi, \beta}$.
When $\beta=0$, this distribution reduces to the uniform distribution over all satisfying assignments $\Omega_n(\Psi)$. For general $\beta\neq 0$, it biases the distribution toward assignments with more $\TRUE$ if $\beta>0$ and biases toward those with more $\FALSE$ if $\beta<0$. 

We consider the following one-shot learning task for  $\beta$. The learner knows parameters $d,k$ and a fixed formula $\Psi\in \Phi_{n,k,d}$.  Additionally, the learner has access to a single sample  $\sigma \in \Omega_n(\Psi)$ drawn from distribution $\Pr_{\Psi,\beta}[\cdot]$. The learner also knows that $\beta$ lies within a specified range $|\beta|\le B$, but it does not know the exact value of $\beta$. The goal is to estimate $\beta$ using these inputs.

To quantify the accuracy of our estimate, we say that $\hat{\beta}$ is an $\epsilon$-estimate if $|\beta-\hat{\beta}|\le \epsilon$. Typically we want $\epsilon$ to  decrease as $n$ increases so that $\epsilon\rightarrow0$ when $n\rightarrow \infty$. In this case we call $\hat{\beta}$ a consistent estimator. 
On the other hand, if there exists a constant $\epsilon_0>0$ such that 
$\limsup_n |\hat{\beta} - \beta| \ge \epsilon_0$, then $\hat{\beta}$ is not a consistent estimator and we say $\beta$ is not identifiable by $\hat{\beta}$. Finally,  if $\beta$ is not identifiable by any $\hat{\beta}$, we say it is impossible to estimate $\beta$.
 
Our main algorithmic result is a linear-time one-shot learning algorithm for $\beta$ in the weighted $k$-SAT model.           \begin{theorem}
\label{thm: algo bound}
    Let $B > 0$ be  a real number. Let $d,k\geq 3$ be integers such that 
    \begin{equation} \label{eq:k/2 condition}
        d\le \frac{1}{e^3\sqrt{k}} \cdot (1+e^{-B})^{\frac{k}{2}}.
    \end{equation} 
    There is an estimation algorithm
    which, for any $\beta^*$ with $|\beta^*| \leq B$,
    given any input $\Phi \in \Phi_{n,k,d}$ and a sample from $\sigma \sim \Pr_{\Phi, \beta^*}$, outputs in $O(n + \log (nB))$ time  an $O(n^{-1/2})$-estimate 
    $\hat{\beta}(\sigma)$ such that 
    \[
    \Pr_{\Phi, \beta^*}\left[
    \big|\hat{\beta}(\sigma) - \beta^*\big| = O(n^{-1/2}) \right] = 1- e^{-\Omega(n)}.
    \]
\end{theorem}

Our results improve upon the conditions in \cite{GKK24}, which ensure a consistent estimate under the requirement when $d\lesssim (1+e^{-B})^{k/6.45}$. Based on the corresponding threshold for approximate sampling, the conjectured “true” threshold for $d$ is of the order $(1+e^{-B})^{\frac{k}{2}}$. Consequently, our improved condition in \eqref{eq:k/2 condition} is only off by a polynomial factor in $k$ relative to this conjectured threshold.

For comparison with the approximate sampling threshold—commonly stated for the uniform $k$-SAT distribution—we specialize to 
$B\rightarrow0$.  In that limit, our algorithmic result for single-sample learning holds roughly when
$d\lesssim 2^{k/2}$.
The best currently known result for efficient sampling, due to \cite{WangYin24}, holds under the condition $d\lesssim 2^{k/4.82}$, see also the series of works \cite{Moitra19, FGYZ21a, JPV21b, HWY23a}. It is conjectured that the sharp condition for efficient sampling is $d\lesssim 2^{k/2}$, supported by matching hardness results for monotone formulas. It is known in particular that for $d\gtrsim 2^{k/2}$, no efficient sampling algorithm exists (unless $\mathsf{NP} = \mathsf{RP}$).

To complement our algorithmic result, we also present impossibility results, suggesting that conditions like \eqref{eq:k/2 condition} are nearly sharp.

\begin{theorem}
   \label{thm: impossible4}
   Let  $\beta^*$ be a real number such that $|\beta^*|>1$. 
       Let  $k \geq 4$ be an even integer, and let $n$ be a multiple of $k/2$ that is large enough. 
    If     \begin{equation}
        \label{eq: e bound}
        d \ge k^3 \left( 1 + \frac{e}{e^{|\beta^*|}-e} \right)^{\frac{k}{2}},
    \end{equation}
    then there exists a formula $\Phi\in \Phi_{n,k,d}$ such that 
    it is impossible to estimate $\beta^*$ from a sample $\sigma \sim \Pr_{\Phi, \beta^*}$ with high probability.
\end{theorem}
For the parameter $d$ around the satisfiability threshold, specifically at the uniquely satisfiable threshold $u(k)$ (see, e.g., \cite{MP10} on the connection between these two thresholds), if
\begin{equation}
    \label{eq:2^k}
    d \ge u(k) =\Theta\Big(\frac{2^{k}}{k}\Big),
\end{equation}
there exists a formula $\Phi\in\Phi_{n,k,d}$ such that it  is impossible to estimate $\beta^*$ from any number of samples $\sigma\sim \Pr_{\Phi, \beta^*}$ 
because $\Omega_n(\Phi)$ is a deterministic set consisting of a single satisfying assignment that does not depend on $\beta^*$. \cite{GKK24} explicitly construct such a formula $\Phi$, though it requires an additional $O(k^2)$ factor relative to \eqref{eq:2^k}, representing the previous best known condition for the impossibility of estimation.
Condition \eqref{eq: e bound} in Theorem~\ref{thm: impossible4} not only relaxes \eqref{eq:2^k} when $|\beta^*|$ grows large, but it also features the correct $k/2$
 exponent, matching that in both \eqref{eq:k/2 condition} and the conjectured threshold.
Indeed, when $B\approx |\beta^*|\rightarrow \infty$, conditions \eqref{eq:k/2 condition} and \eqref{eq: e bound} both take the form 
\begin{equation}
    (1+O(e^{-|\beta^*|}))^{k/2}\cdot k^{O(1)}
\end{equation}
These findings partially indicate that, at least for the 
$k$-SAT model, the sampling threshold is more relevant to one-shot learning than the satisfiability threshold.

In addition, we find that if we allow 
 $\beta^*$  to be proportional to  $k$ , then learning becomes impossible for a significantly larger range of $d$.
Specifically, unlike condition \eqref{eq: e bound}, which requires  $d$ to be exponential in $k$, here we only need 
$d$ to be quadratic in $k$, leading to a much sparser formula when $k$ is large.
\begin{theorem}
\label{thm:impossible2}
Let $k\geq 4$ be an even integer. For all $\beta^*\in\mathbb{R}$ such that 
         $|\beta^*| \ge k\ln2 $ the following holds. Let $n$ be a multiple of $k/2$ that is large enough. If
     $d\ge k^2/2$, then there exists a formula $\Phi\in \Phi_{n,k,d}$ such that 
    it is impossible to estimate $\beta^*$ from a sample $\sigma \sim \Pr_{\Phi, \beta^*}$ with high probability.
\end{theorem}

\begin{remark}
In the regimes where 
Theorem~\ref{thm: impossible4} or Theorem~\ref{thm:impossible2} apply, 
the corresponding formula $\Phi$ ensures that there is a single assignment that is the output with all but exponentially-small probability, regardless of the value of~$\beta^*$.
Hence, the proof of this theorem guarantees that it is impossible to learn from  exponentially many independent samples. 
Moreover, for any pair of $(\beta_1, \beta_2)$ such that $|\beta_1| > |\beta_2|\ge |\beta^*|$ and $\beta_1\beta_2\ge0$, no hypothesis testing for $H_0: \beta =\beta_1$ versus $ H_1: \beta =\beta_2$ can be done to distinguish  $\Pr_{\Phi, \beta_1}$ from  $\Pr_{\Phi, \beta_2}$, that is, there exists no sequence of consistent test functions $\phi_n:\Omega_n \rightarrow\{0,1\}$ such that 
\[
\lim_{n\rightarrow\infty}\E_{\sigma\sim \Pr_{\Phi, \beta_1}} \phi_n(\sigma) = 0 \text{ and } \lim_{n\rightarrow\infty}\E_{\sigma\sim \Pr_{\Phi, \beta_2}} \phi_n(\sigma) = 1.
\]
\end{remark}

\subsection{Proof overview}
We prove Theorem~\ref{thm: algo bound} by using the maximum pseudo-likelihood estimator. 
Establishing the consistency of this estimator—as stated in Theorem~\ref{thm: consistent}—requires demonstrating that the log-pseudo-likelihood function is (strongly) concave; see \eqref{eq: GKK bound 2} for the precise formulation.  In the 
$k$-SAT setting, showing such concavity amounts to showing that with high probability over samples drawn from $\Pr_{\Phi,\beta^*}[\cdot]$, each sample contains a linear number of “flippable variables.” 
We prove this property in Lemma~\ref{lem: flippable} by applying the  Lov\'asz local lemma (LLL), which enables us to compare $\Pr_{\Phi,\beta^*}[\cdot]$  to a suitable product distribution—under which the number of flippable variables is guaranteed to be linear. Notably, we apply the LLL directly to a non-local set of variables, in contrast to previous analyses that confined the application of the LLL to local neighbourhoods—a restriction that typically imposes stronger constraints on the parameter regime. By circumventing these stronger constraints, our approach achieves its guarantee under the nearly optimal LLL condition.

The main technical novelty of this paper is our negative results. To explain these,  we begin by outlining the proof of Theorem~\ref{thm:impossible2}. For simplicity, we focus on the case where  $\beta^*\ge 0$. At a high level, we construct a gadget formula $\Psi_0$ for which the all-true assignment $\sigma^+$ carries almost all of the probability mass under $ \Pr_{\Psi_0,\beta^*}[\cdot]$ provided that $\beta^*\ge k \ln 2$. Consequently, a sample $\sigma\sim \Pr_{\Psi_0,\beta^*}[\cdot]$ drawn from this distribution is nearly deterministic, offering virtually no information about $\beta^*$
  and thus rendering learning impossible.
The key property of this gadget is established in Lemma~\ref{lem: gadget}. 
Specifically, the lemma ensures that $\sigma^+$ satisfies $\Psi_0$ and that any other assignment $\sigma$ with 
fewer than $2n/k$ variables set to $\FALSE$ is not a satisfying assignment of $\Psi_0$. 
We achieve this by incorporating a cyclic structure over the variables that enforces global correlation among the $\FALSE$ values in the assignments. In particular, there are no flippable variables in $\sigma^+$.

Towards the proof of Theorem~\ref{thm: impossible4},
we first leverage the gadget $\Psi_0$ to show the existence of a stronger gadget $\Psi_2$, parametrized by $b>1$ in Lemma~\ref{lem: large gadget}, which guarantees that the all-true assignment $\sigma^+$ satisfies $\Psi_2$ and any other assignment 
with fewer than $n/b$ variables set to $\FALSE$  fails to satisfy~$\Psi_2$. 
Then we choose $b$ appropriately in terms of $\beta^*$ to make sure that $\sigma^+$ carries nearly all of the probability mass, using some more technical estimates for the corresponding partition function.
To build $\Psi_2$, we take a finite number of replicas of $\Psi_0$ on randomly permuted sets of variables. The existence of the desired formula is established using the probabilistic method; specifically, we demonstrate an upper bound on the expectation of $m(\sigma)$ for any satisfying assignment $\sigma\neq \sigma^+$, over the choice of the permutations.

\subsection{Related work}
\noindent \textbf{Parameter Estimation in Markov Random Fields.}  A large body of work has focused on parameter estimation under the one-shot learning paradigm (see, e.g., \cite{Chatterjee07c, BM18, DDK19, DDP20, GM20, DDDVK21, MSB}), particularly for Ising-like models in statistical physics and for dependent regression models in statistics. In this work, we follow a similar approach by establishing the consistency of the maximum pseudo-likelihood estimator.  Earlier studies (e.g., \cite{Gidas88, CometsGidas91, Comets92, GeyerThompson92}) have also explored parameter estimation in Markov random fields using the maximum likelihood estimator.

Before our work, the papers~\cite{BR21,GKK24} were the first to study one-shot learning in hard-constrained models. In particular, the hardcore model analysed in \cite{BR21} can be viewed as a weighted monotone $2$-SAT model, and one natural extension of the hardcore model to $k$-uniform hypergraphs corresponds to the class of weighted monotone $k$-SAT models—a special case of the weighted $k$-SAT models that we consider. Because a typical assignment in these monotone formulas possesses $\Omega(n)$ flippable variables, the pseudo-likelihood estimator remains consistent across all parameter regimes, and no phase transition is expected. 
The weighted $k$-SAT problem was analysed in \cite{GKK24}, where the authors derived both a consistency condition and an impossibility condition, though a substantial gap remained between them. By tightening the bounds on both ends, our work considerably narrows this gap, nearly closing it entirely.
\medskip

\noindent\textbf{Related Works in Structural Learning/Testing.}
An alternative direction in learning Markov random fields involves estimating the interaction matrix between variables—a question originally posed by \cite{ChowLiu68}. For the Ising model, this problem has been extensively studied (see, e.g., \cite{Bresler15, VMLC16, BGS17} and the references therein), and subsequent work has extended the results to higher-order models \cite{KilvansMeka17, HKM17, GMM24}. Recent work \cite{ZhangKKW20, DDDVK21, GM24} has also considered the joint learning of both structure and parameters. Moreover, \cite{SanthanamWainwright12} establishes the information-theoretic limits on what any learner can achieve, and similar analyses have been conducted for hardcore models \cite{BGS14a, BGS14b}.
While some approaches in this line of work require multiple independent samples, as noted in \cite{DDDVK21}, it is also possible to reduce learning with $O(1)$ samples to a class of special cases within one-shot learning. 
Related problems in one-shot testing for Markov random fields have also been studied in \cite{BreslerNagaraj18,DDK18,MukherjeeMukherjeeYuan18,BBCSV20,BCSV21}.

\subsection{$k$-SAT Notation}
\label{sec:prelim}
We use standard notation for the $k$-SAT problem. A formula $\Phi=(V,C)$ denotes a  \emph{CNF (Conjunctive Normal Form) formula} with variables   $V=\{x_1,\dots, x_n\}$ and clauses~$C$. We use $\sigma(x_i)$ and $\sigma_i$ to denote the truth value of the variable~$x_i$ under an assignment $\sigma \colon V \to \{\TRUE,\FALSE\}$.
For any clause $c\in C$,   $var(c)$ denotes the set of variables appearing in $c$ (negated or not).  
The \emph{degree} of a variable~$x_j$ in~$\Phi$ is   the number of clauses in which $x_j$ or $\neg x_j$ appears,  namely $|\{c\in C: x_j\in var(c)\}|$. The degree of~$\Phi$ is the maximum, over $x_j \in V$, of the degree of $x_j$ in~$\Phi$.   
As noted in the introduction, $m(\sigma)$
denotes the number of variables that are assigned to~$\TRUE$ by an assignment~$\sigma$. Namely, 
$m(\sigma):=|\{i\in [n]:\sigma_i = \TRUE\}|$.

\section{Maximum pseudo-likelihood estimator: the proof of Theorem~\ref{thm: algo bound}}
Section~\ref{sec: MPLE} introduces the fundamentals of the maximum pseudo-likelihood estimator and analyses its running time for solving the weighted $k$-SAT problem. 
In Sections~\ref{sec: consistency} and \ref{sec: flippable},
we establish the estimator’s consistency.

\subsection{Overview of maximum (pseudo)-likelihood estimation}
\label{sec: MPLE}
For a $k$-SAT formula $\Psi$ and a satisfying assignment~$\sigma$, we will use $f(\beta;\sigma)$ to denote the quantity~$\Pr_{\Psi,\beta}(\sigma)$ from \eqref{eq:model}, i.e., $f(\beta; \sigma) \;=\; 
e^{\beta \,m(\sigma)}/Z(\beta; \Psi)$, where
 $Z(\beta;\Psi)$ is the normalising constant of the distribution (the partition function).

A standard approach to parameter estimation is to find
$
    \hat{\beta}_{\mathrm{MLE}}(\sigma) 
    \;:=\; \arg\max_{\beta}\, f(\beta; \sigma)$,
which is commonly referred to as the maximum likelihood estimate (MLE). However, two main obstacles arise when applying MLE directly to the weighted \(k\)-SAT problem. First, (approximately)
computing \(Z(\beta; \Psi)\) is generally intractable because it is an NP-hard computation. Second, even if an approximation algorithm exists for computing \(\hat{\beta}_{\mathrm{MLE}}(\sigma)\), there is no guarantee of its consistency, i.e., there is no guarantee that with high probability it is close to $\beta^*$.
Hence, we take a  computationally more tractable variant of MLE from \cite{Besag74}  which is called the maximum pseudo-likelihood estimation. 
Let $f_i(\beta; \sigma)$ be the conditional probability of $\sigma_i$,  in a distribution with parameter \(\beta\), conditioned on the value of $\sigma_{-i}:=(\sigma_j)_{j\neq i}$. 
The maximum pseudo-likelihood estimate (MPLE) is defined as
\[
    \hat{\beta}_{\mathrm{MPLE}}(\sigma) :=\; \arg\max_{\beta}\, \prod_{i\in V} f_i(\beta; \sigma)
    =\; \arg\max_{\beta}\, \sum_{i\in V} \ln f_i(\beta; \sigma).
\]

Here, the objective function $F(\beta;\sigma) :=\sum_{i\in V} \ln f_i(\beta; \sigma)$ is the so-called \emph{log-pseudo-likelihood function}. 
For the weighted $k$-SAT problem, it is not hard to compute $f_i(\beta; \sigma)$ as
\[
f_i(\beta;\sigma) = \frac{e^{\beta \sigma_i}}{e^\beta \mathbbm{1}[\sigma_{-i} \wedge (\sigma_i \leftarrow \TRUE)] + \mathbbm{1}[\sigma_{-i} \wedge (\sigma_i \leftarrow \FALSE)]},
\]
where  $\mathbbm{1}[\omega]$ is shorthand for 
$\mathbbm{1}[\omega \models \Psi]$.
So we can write the log-pseudo-likelihood function for $k$-SAT as
\begin{align*}
     F(\beta;\sigma)  &=
\sum_{i\in V} \ln \left( \frac{e^{\beta \sigma_i}}{e^\beta \mathbbm{1}[\sigma_{-i} \wedge (\sigma_i \leftarrow \TRUE)] + \mathbbm{1}[\sigma_{-i} \wedge (\sigma_i \leftarrow \FALSE)]} \right), \\
&= \beta m(\sigma) -\sum_{i\in V} \ln \left( e^\beta \mathbbm{1}[\sigma_{-i} \wedge (\sigma_i \leftarrow \TRUE)] + \mathbbm{1}[\sigma_{-i} \wedge (\sigma_i \leftarrow \FALSE)] \right).
\end{align*}
For a fixed $\sigma$, 
$F(\cdot;\sigma):\mathbb{R}\rightarrow\mathbb{R}$ is a function of $\beta$.
By taking derivative with respect to $\beta$, we obtain 
\begin{equation}
    \label{eq:derivative F}
    \frac{\partial F(\beta;\sigma)}{\partial \beta} =
    m(\sigma) - \sum_{i\in V} \frac{e^\beta  \mathbbm{1}[\sigma_{-i} \wedge (\sigma_i \leftarrow \TRUE)] }{ e^\beta \mathbbm{1}[\sigma_{-i} \wedge (\sigma_i \leftarrow \TRUE)] + \mathbbm{1}[\sigma_{-i} \wedge (\sigma_i \leftarrow \FALSE)]}. 
\end{equation}
Clearly, $\frac{\partial F(\beta;\sigma)}{\partial \beta}$ is a decreasing function of $\beta$, which implies that $F(\beta;\sigma)$ has a unique global maximum  that is achieved when $\frac{\partial F(\beta;\sigma)}{\partial \beta}=0$.
Therefore, $ \hat{\beta}_{\mathrm{MPLE}}(\sigma)$ can be uniquely defined to be the maximum of $F(\beta;\sigma)$. 

Moreover, provided $|\hat{\beta}_{\mathrm{MPLE}}(\sigma)|\le 2B$, an  $\epsilon$-close estimate of $\hat{\beta}_{\mathrm{MPLE}}(\sigma)$ can be computed, using $O(\ln{(B/\epsilon)})$ steps of binary search for the solution to $\frac{\partial F(\beta;\sigma)}{\partial \beta}=0$. At each step of the binary search, we evaluate $\frac{\partial F(\beta;\sigma)}{\partial \beta}$ and adjust the binary search interval based on its sign.
A naive evaluation of $\frac{\partial F(\beta;\sigma)}{\partial \beta}$ as in \eqref{eq:derivative F} would require $\Theta(n)$ operations  per step. We can reduce this by exploiting the fact that the summand
\[
S_i(\beta):=\frac{e^\beta  \mathbbm{1}[\sigma_{-i} \wedge (\sigma_i \leftarrow \TRUE)] }{ e^\beta \mathbbm{1}[\sigma_{-i} \wedge (\sigma_i \leftarrow \TRUE)] + \mathbbm{1}[\sigma_{-i} \wedge (\sigma_i \leftarrow \FALSE)]}
\] 
can take only three values $\{0,1, e^{\beta}/(1+e^{\beta})\}$. Hence, by grouping the summands according to their values, we obtain
\begin{equation}
    \label{eq: derivative F2}
 m(\sigma) - \sum_{i\in V} S_i( \beta)
 = m(\sigma) -  |\{i\in V: S_i(\beta) = 1\}|- \frac{e^\beta}{1+e^\beta} \cdot \left| \left\{i\in V: S_i(\beta) = \frac{e^\beta}{1+e^\beta} \right\} \right|.
\end{equation}
Crucially, the sets $\{i\in V: S_i(\beta) = 1\}$ and $\{i\in V: S_i(\beta) = \frac{e^\beta}{1+e^\beta} \}$ do not depend on $\beta$, and thus can be computed only once before the binary search.
After this preprocessing, each evaluation of $\frac{\partial F(\beta;\sigma)}{\partial \beta}$
  can be done in $O(1)$ time using the decomposition in \eqref{eq: derivative F2}. Overall, to achieve an $O(n^{-1/2})$-close estimate, the total running time is 
  $O(n+ \log (nB)).$

\subsection{Consistency of MPLE}
\label{sec: consistency}
In Section~\ref{sec: MPLE}  we gave an algorithm with running time $O(n + \log(nB) )$, which 
takes as input a formula $\Psi$ and an assignment $\sigma\in\Omega_n(\Psi)$ and outputs an estimate $\hat{\beta}(\sigma)$ 
which  satisfies
\[
\left| \hat{\beta}(\sigma) -\hat{\beta}_{\mathrm{MPLE}}(\sigma) \right| = O(n^{-1/2}),
\]
provided $|\hat{\beta}_{\mathrm{MPLE}}(\sigma)|\le 2B$.
Theorem~\ref{thm: algo bound} follows immediately from the 
Theorem~\ref{thm: consistent}, which demonstrates
$O(\sqrt{n})$-consistency.
\begin{theorem}
    \label{thm: consistent} Let $\beta^*$ be a real number, and let $d,k\geq 3$ be integers such that 
    \begin{equation}
        \label{eq:k/2 condition 2}
        d\le \frac{1}{e^3\sqrt{k}} \cdot (1+e^{-|\beta^*|})^{\frac{k}{2}}.
    \end{equation} 
    For any integer $n$ and $\Phi \in \Phi_{n,k,d}$, 
    \begin{equation}
        \label{eq:consistency}
         \Pr_{\Phi, \beta^*} \left[ 
            \left| 
                \hat{\beta}_{\mathrm{MPLE}}(\sigma) -\beta^*
            \right| = O(n^{-1/2})
         \right] = 1- e^{-\Omega(n)}.
    \end{equation}
\end{theorem}

As in the standard literature \cite{Chatterjee07c}, 
the consistency result can be proved using bounds on the derivatives of the log-pseudo-likelihood function $F$.
In particular, following  the analysis in \cite{GKK24}, 
\eqref{eq:consistency} is a consequence of these two bounds:
\begin{enumerate}
    \item A uniform linear upper bound of the second moment of the first derivative of $F$: for all $\beta\in \mathbb{R}$,
    \begin{equation}
        \label{eq: GKK bound 1}
    \E_{\Phi, \beta} \left[
        \left(
            \frac{\partial F(\beta;\sigma)}{\partial \beta}
        \right)^2
        \right] \le kdn.
    \end{equation}
    \item 
    With high probability over $\sigma\sim  \Pr_{\Phi, \beta^*}$, 
    there is a uniform lower bound of the second derivative of $F$:
    \begin{equation}
        \label{eq: GKK bound 2}
          \Pr_{\Phi, \beta^*} \left[
          \inf_{\beta\in \mathrm{R}} 
            \frac{\partial^2 F(\beta;\sigma)}{\partial \beta^2}
            = \Omega(n)
          \right] = 1- e^{-\Omega(n)}.
    \end{equation}
\end{enumerate}
In brief, \eqref{eq: GKK bound 1} and \eqref{eq: GKK bound 2} control the first and second-derivative of the log–pseudo-likelihood function, for most $\sigma$. A second-order Taylor approximation around $\beta^*$, combined with these bounds, yields \eqref{eq:consistency}. We refer interested readers to the proof of Theorem 1.1 in \cite{GKK24}, where the complete argument is carried out in detail.

Moreover,  Lemma 3.1 in \cite{GKK24} proves the bound \eqref{eq: GKK bound 1} for all $\beta\in \mathbb{R}$ and all $d,k\ge 3$, and they give a combinatorial expression (equations (3.9) and (3.10) in \cite{GKK24}) for $ \frac{\partial^2 F(\beta;\sigma)}{\partial \beta^2}$ that will be the useful for proving \eqref{eq: GKK bound 2}.
To state this expression, we introduce the notion of flippable variables. We say a variable $v_i$ is \emph{flippable} in $\sigma$ if the assignment, obtained by flipping the value of variable $v_i$ while keeping the values of other variables in $\sigma$, is still a satisfying assignment, that is, $(\sigma_{-i} \wedge ( \neg \sigma_i)) \models \Psi$.  
We use
$e_{v_i}(\sigma)$ to denote the indicator of the event that variable $v_i$ is flippable in a satisfying assignment $\sigma$. By differentiating \eqref{eq:derivative F} using the expression in \eqref{eq: derivative F2}, we obtain the following expression for the second derivative of $F$ (shown in \cite{GKK24}):
\begin{equation}
    \label{eq: GKK bound 3}
      \frac{\partial^2 F(\beta;\sigma)}{\partial \beta^2} = \frac{e^\beta}{(e^\beta+1)^2} \sum_{v\in V} e_v(\sigma).
\end{equation}
Hence, proving \eqref{eq: GKK bound 2} reduces to establishing a linear lower bound on the number of flippable variables. The main 
ingredient in the proof of our positive result
is  Lemma~\ref{lem: flippable}, which provides such a lower bound under the condition \eqref{eq:k/2 condition 2}.

\begin{lemma}
    \label{lem: flippable}
Let $\beta^*$ be a real number, and let $d,k\geq 3$ be integers such that 
    \begin{equation*}
        d\le \frac{1}{e^3\sqrt{k}} \cdot (1+e^{-|\beta^*|})^{\frac{k}{2}}.
    \end{equation*}
    Then for a fixed $\Phi=(V,C)\in \Phi_{n,k,d}$, 
    \begin{equation}\label{eq: many flippable}
         \Pr_{\Phi, \beta^*} \left[\sum_{v\in V} e_v(\sigma) = \Omega(n)\right] = 1- e^{-\Omega(n)}.
    \end{equation}
\end{lemma}
Using Lemma~\ref{lem: flippable} and the identity \eqref{eq: GKK bound 3}, we derive \eqref{eq: GKK bound 2} under the condition \eqref{eq:k/2 condition 2}, thereby completing the proof of Theorem~\ref{thm: consistent}. The proof of Lemma~\ref{lem: flippable} will be presented in the next section.

\subsection{Proof of Lemma~\ref{lem: flippable}: applying the Lov\'asz local lemma in a batch} 
\label{sec: flippable}

The following version of the Lov\'asz Local Lemma (LLL) from \cite{GJL19} will be useful for our proof of Lemma~\ref{lem: flippable}.

\begin{lemma}[Lemma 26, \cite{GJL19}]
\label{lem:LLL GJL19}\label{lem:L}
    Let $\Phi=(V,C)$ be a CNF formula.
    Let $\mu$ be the product distribution on $\{\TRUE,\FALSE\}^{|V|}$, where each variable is set to $\TRUE$ independently with probability $e^\beta/(1+e^\beta)$.
    For each $c \in C$, let $A_c$ be the event that clause $c$ is unsatisfied.
    If there exists a sequence $\{x_c\}_{c\in C}$ such that 
    for each $c \in C$,
    \begin{equation}
        \label{eq:LLL condition}
        \Pr_\mu[A_c] \le x_c \cdot \prod_{j\in \Gamma(c)} (1 - x_j), 
    \end{equation}
    where $\Gamma(c) \subseteq C$ is the set of clauses that contain a variable in $var(c)$, then $\Phi$ has a satisfying assignment. 
    Moreover, the distributions $\Pr_{\Phi, \beta}$ and $\Pr_\mu$ can be related as follows: for any event $E$ that can be completely determined by the assignment of a set $S$ of variables, 
    \begin{equation}
        \label{eq: LLL translation}
        \Pr_{\Phi, \beta}[E] \le \Pr_\mu[E] \cdot \prod_{j \in \Gamma(E)} \frac{1}{1 - x_j},
    \end{equation}
    where $\Gamma(E)$ denotes the set of all clauses that contain a variable in $S$.
\end{lemma}
In many previous works utilising the LLL for sampling purposes, Lemma~\ref{lem:L}, or a variant of it, is typically applied to  local events, including, for instance, the event that a specific variable  is flippable or, more generally, to events happening in the neighbourhood of a vertex. This is present in the approach of \cite{GKK24} which, in the end, imposed stricter conditions on $d$ (relative to $k$) since it requires ``marking'' variables appropriately (see \cite{Moitra19}).
% This local approach often requires additional techniques, such as correlation decay \cite{WangYin24} or Moitra’s marking \cite{Moitra19}, to handle the behaviour of multiple variables. 
Here, we prove Lemma~\ref{lem: flippable} by applying Lemma~\ref{lem:LLL GJL19} directly to a \emph{batch} of random variables scattered around the graph in one go, removing the need for marking, and relaxing significantly the conditions on $d$. This simple idea enables our stronger result.
% Although this idea is simple, we believe it may have broader applications in the study of the LLL.

\begin{proof}[Proof of Lemma~\ref{lem: flippable}]
Let $\beta = |\beta^*|$ and let $V=\{v_1,\dots, v_n\}$.
    For any $c\in C$, we set $x_c = 1/(d^2k+1)$. Note $\Gamma(c)\le dk$.
    By 
\eqref{eq:k/2 condition 2}
and the trivial bound 
$1\le d^2 k$,  
we have
    \[d^2 k+1 \le  
    \frac{1}{2e} (1+e^{-\beta})^k =
    \frac{1}{2e} \left( \frac{1+e^{\beta}}{e^{\beta}} \right)^k. \]
    Also, since $d,k\ge 3$, we have 
    \[
    \frac{1}{2e} \le  \left(1-\frac{1}{dk} \right)^{dk}.
    \]
    Thus, for each $c\in C$,
    \begin{align*}
        \Pr_\mu[A_c] &\le \left( \frac{e^\beta}{1+e^\beta} \right)^k 
        \le \frac{1}{2e} \cdot \frac{1}{d^2k+1} 
        \le  \left(1-\frac{1}{dk} \right)^{dk} \frac{1}{d^2k+1} \\
        &\le  \left(1-\frac{1}{d^2k+1} \right)^{dk} \frac{1}{d^2k+1}
        \le x_c \cdot \prod_{j\in \Gamma(c)} (1 - x_j),
    \end{align*}
    establishing condition \eqref{eq:LLL condition} in Lemma~\ref{lem:LLL GJL19}.
    
    Next we will show \eqref{eq: many flippable}   under condition \eqref{eq:LLL condition}. 
    Recall for an assignment $\sigma$ of $\Omega_n(\Phi)$ and a variable $v_i\in V$, $e_{v_i}(\sigma)=1$ if in every clause $c$ containing $v_i$, there is a variable $v_j\neq v_i$ that satisfies $c$. 
    Since $d$ and $k$ are bounded, for sufficiently large $n$, 
    there exists a set $U$ of $R=\Omega(n)$ variables $v_1, \dots, v_R$ such that
    $d(v_i, v_j) \ge 100$ for all $1\le i<j \le R$, where distance $d(\cdot, \cdot)$ is defined as the graphical distance in the hypergraph corresponding to $\Phi$.
    Let $E$ denote the event $\{\sum_{i=1}^R e_{v_i}(\sigma) < \frac{R}{3}\}$.

    First, we compute the probability of $E$ under the product distribution $\mu$, where each variable is set to $\TRUE$ independently with probability $e^\beta/(1+e^\beta)$. 
    We apply a union bound by noting that if $E$ occurs then there are at least $\frac{2R}{3}$ variables in $U$ that are not flippable.
    \begin{equation}
        \label{eq: union bound 1}
        \Pr_\mu[E]  = \Pr_\mu\bigg[\sum_{i=1}^R e_{v_i}(\sigma) < \frac{R}{3}\bigg] \le \binom{R}{2R/3} \left( \max_{v_i\in U} \Pr_\mu[e_{v_i}(\sigma) = 0] \right)^{2R/3}.
    \end{equation}
    If $e_{v_i}(\sigma) = 0$, then there exists a clause $c_j$ such that $v_i\in var(c_j)$ and $c_j$ is not satisfied by $var(c_j) \setminus \{v_i\}$ in $\sigma$.
    We apply another union bound over all $c_j$ in which $v_i$ appears, and obtain
    \begin{equation}
        \label{eq: union bound 2}
        \Pr_\mu[e_{v_i}(\sigma) = 0] \le d\cdot \left( \frac{e^\beta}{1+e^\beta} \right)^{k-1}.
    \end{equation}
    From \eqref{eq: union bound 1} and \eqref{eq: union bound 2}, we have
    \begin{equation}
        \label{eq: product prob E}
        \Pr_\mu [E] \le  \binom{R}{2R/3} 
        \bigg[d\cdot \left( \frac{e^\beta}{1+e^\beta} \right)^{k-1}  \bigg]^{2R/3}.
    \end{equation}

We now apply Lemma~\ref{lem:LLL GJL19}  to relate the distribution $\Pr_{\Phi, \beta}[\cdot]$  to $\Pr_\mu[\cdot]$. 
Let $S$ be the set of variables  that are either in $U$ or share a clause with variables in $U$.  Then    $E$ is determined by   the variables $S$. 
So $\Gamma(E)$ is the set of all clauses containing variables in $S$, and thus $|\Gamma(E)| \le Rd^2k$.
    It follows from  \eqref{eq: LLL translation}, \eqref{eq: product prob E}, and the standard bound $\binom{n}{m} \leq \big(\tfrac{ne}{m}\big)^m$
    that
    \begin{align*}
                \Pr_{\Phi, \beta}[E] &\le \Pr_\mu[E] \cdot \left( 1 - \frac{1}{d^2 k+1} \right)^{-|\Gamma(E)|}  \\
                &\le  \binom{R}{2R/3} 
        \bigg[d\cdot \left( \frac{e^\beta}{1+e^\beta} \right)^{k-1}  
        \bigg]^{2R/3} \left( 1 - \frac{1}{d^2 k+1} \right)^{-Rd^2k} \\
        &\le 
(3e/2)^{2R/3}
\bigg[d\cdot \left( \frac{e^\beta}{1+e^\beta} \right)^{k-1}  
        \bigg]^{2R/3} e^R 
        \le \bigg[ e^3 \cdot d \cdot \left( \frac{e^\beta}{1+e^\beta} \right)^{k-1}  \bigg]^{2R/3}.
    \end{align*}
    This completes the proof since we have $
        e^3 \cdot d \cdot \big( \frac{e^\beta}{1+e^\beta} \big)^{k-1} <1$ by our assumption \eqref{eq:k/2 condition 2}.  
\end{proof}

\section{Impossibility of learning: proofs of Theorems~\ref{thm: impossible4} and~\ref{thm:impossible2} } 
For the construction of the impossibility instances, we begin with a ``gadget'' that will serve as a building block.

\begin{lemma}\label{lem: gadget}
Let $k\ge 4$ be an even integer and let $n\geq k$ be a multiple of $k/2$. If $d\ge k^2/2$, then there is a formula $\Psi_0\in\Phi_{n,k,d}$  such that if an assignment $\sigma$ of $\Psi_0$   is satisfying then either 
    \begin{enumerate}
        \item $\sigma$ has $n$ $\TRUE$s, or
        \item $\sigma$ has at least $2n/k$ $\FALSE$s.  
    \end{enumerate}
Similarly (by symmetry), there is a formula $\Psi_1 \in \Phi_{n,k,d}$ such that, for any satisfying assignment $\sigma$ of $\Psi_1$, either   $\sigma$ has $n$  $\FALSE$s, or $\sigma$ has at least $2n/k$  $\TRUE$s.   
\end{lemma}
The instance~$\Psi_0$ and~$\Psi_1$ will lead to a proof of Theorem~\ref{thm:impossible2}.
For clarity, we denote the assignment to $n$ variables with 
$n$ $\TRUE$s as $\sigma^+$.
\begin{proof}[Proof of Theorem~\ref{thm:impossible2}]
    Assume $\beta^* \ge k \ln 2$ without loss of generality.
    Let $\Psi_0$ be the formula given by Lemma~\ref{lem: gadget} (when $\beta^* \leq -k\ln 2$ we use $\Psi_1$ instead).
    Suppose $\sigma$ is drawn from $\Pr_{\Psi_0, \beta^*}$, and we directly estimate the probability of $\{\sigma = \sigma^+\}$:
    since $\sigma$ has at most $2^n$ possibilities and on event $\{\sigma\neq \sigma^+\}$ the number of $\TRUE$s is at most $n- 2n/k$, we have
    \begin{align*}
        \Pr_{\Psi_0, \beta^*}[\sigma = \sigma^+] \ge 
        \frac{e^{\beta^*n}}{e^{\beta^*n} + 2^n \cdot e^{\beta^* \cdot  (n - 2n/k)}} \ge \frac{2^{kn}}{2^{kn} + 2^{n+k\cdot(n-2n/k)}}
        = \frac{1}{1+2^{-n}}.
    \end{align*}
    As the samples from $\Pr_{\Psi_0, \beta^*}$ are insensitive to $\beta^*$ with high probability, learning $\beta^*$ from $\sigma$ is impossible.
\end{proof}

We now present a detailed description of the formula that defines $\Psi_0$ in Lemma~\ref{lem: gadget}. 
Let $k\geq 3$ be an even integer and let $n\geq k$ be a multiple of $k/2$.  
Let $N = \{0,\ldots,n-1\}$.
The variables of~$\Psi_0$ are $\{ x_i \mid i\in N\}$.

For the construction, it will be helpful to group variables in batches of $k/2$ variables in a cyclic manner. Specifically, for $i\in  \mathbb{N}$, consider the $i$-th batch of indices
\[\Xi_i = \{ i+j \pmod n \mid j \in \{0,\ldots, k/2-1\} \}\] 
and let $C_i = \{x_\ell \mid \ell \in \Xi_i\}$ be the corresponding variables in the $i$-th batch. We now introduce two types of length-$k/2$ clauses $W_{i,\ell}$ and $\Pi_i$ that will be used to form the final length-$k$ clauses. 
Specifically, for each $\ell$ in the batch $\Xi_i$, $W_{i,\ell}$ is the length-$(k/2)$ clause with variable set~$C_i$ in which $x_\ell$ appears positively and all other variables are negated. 
Let $\Pi_i$ be the length-$(k/2)$ clause with variable set~$C_i$ in which all variables are negated. 
Finally, 
\[\Psi_0 := \bigwedge_{i\in N, \ell \in \Xi_i} (W_{i,\ell} \vee \Pi_{i+k/2}).\]
(so
 $\Psi_0$
is the formula with variable set $  \{ x_i \mid i \in N\}$ 
and clause set 
$  \{ W_{i,\ell} \vee \Pi_{i+k/2} \mid  {i \in  N}, {\ell \in \Xi_{i}}  \}  
 $).
The formula $\Psi_1$ is obtained from $\Psi_0$ by negating all of the literals.

Note that $\Psi_0, \Psi_1 \in \Phi_{n,d,k}$ for every integer $d\geq k^2/2$,  
since for each $j\in N$, the literals~$x_j$ and $\neg x_j$ occur (together) $k^2/2$  times. To see this, 
note that for each $j\in N$, $x_j$ or $\neg x_j$ comes up in $W_{i,\ell} \vee \Pi_{i+k/2}$ for 
all $i\in \{j-k-1 \pmod n,\ldots,j\}$ and all $\ell \in  \Xi_i$ so this is $k$ different $i$'s and $k/2$ different $\ell$'s.  
In the proof of Lemma~\ref{lem:CoverGadget} we will demonstrate that  $\Psi_0$ (and analogously, $\Psi_1$) satisfies the requirements of Lemma~\ref{lem: gadget}.

\begin{lemma}\label{lem:CoverGadget}
Let $k\ge 4$ be an even integer and let $n\geq k$ be a multiple of $k/2$. 
If $\sigma \neq \sigma^+$ satisfies $\Psi_0$
then,  for all 
$\ell \in \{0,1,\ldots, 2n/k-1\}$, 
$\sigma$ assigns at least one variable in~$C_{\ell k/2} \cup C_{(\ell+1)k/2}$ to $\FALSE$ and $\sigma$ has at least $2n/k$ $\FALSE$s in total.
\end{lemma}

\begin{proof}[Proof of Lemma~\ref{lem:CoverGadget}]
We will use the following claim as a key step of the proof.

{\bf Claim:} If  $\sigma\neq\sigma^+$ satisfies $\Psi_0$, and $\sigma$ assigns one variable $x_a$ in $C_{\ell k/2}$ to $\FALSE$, then either
\begin{enumerate}
    \item $\sigma$ assigns a variable $x_b$ in $C_{(\ell+1)k/2}$ to $\FALSE$, or
    \item All variables in $C_{(\ell+1)k/2}$ are assigned to $\TRUE$ in $\sigma$, and there exist variables $x_c\in C_{(\ell +2)k/2}$ and $x_d\in C_{\ell k/2}\setminus\{x_a\}$ that are assigned to $\FALSE$ in $\sigma$.
\end{enumerate}

{\bf Proof of the Claim:} 
Suppose we are not in the first case, so we will show that $\sigma$ assigns $\FALSE$ to $x_c\in C_{(\ell +2)k/2}$ and $x_d\in C_{\ell k/2}\setminus\{x_a\}$. 
Since $\sigma$ satisfies $\Psi_0$, 
% it satisfies $\psi_{\ell k/2}$ so 
it satisfies at least one of $W_{\ell k/2,a}$ and $\Pi_{(\ell+1)k/2}$. 
Since variables in $C_{(\ell+1)k/2}$ are all assigned to $\TRUE$ by the assumption that we are not in the first case, $\sigma$ does not satisfy $\Pi_{(\ell+1)k/2}$. 
If $\sigma$ satisfies $W_{\ell k/2,a}$ then it assigns a variable 
$x_{d}$ to $\FALSE$ where $d \in \Xi_{\ell k/2} \setminus \{a\}$.
Also, the set $\{j\in \Xi_{\ell k/2}: \sigma(x_j)=\FALSE\}$ is not empty.
Let $j=\max\{j\in \Xi_{\ell k/2}: \sigma(x_j)=\FALSE\}$.
Note that $\sigma$ does not satisfy $W_{j,j}$, so for $\sigma$ to satisfy $W_{j,j} \vee \Pi_{j+(k/2)}$, it must satisfy $\Pi_{j+(k/2)}$, which means $\sigma$
assigns a variable $x_c\in C_{j+(k/2)}$ to $\FALSE$.
Since $x_c\notin C_{(\ell+1)k/2}$ and $C_{j+(k/2)}\subseteq C_{(\ell+1)k/2} \cup C_{(\ell+2)k/2}$, we establish that $x_c\in  C_{(\ell+2)k/2}$. This concludes the proof of the claim.

We now show how to use the claim to prove the lemma. 
Fix an assignment $\sigma \neq \sigma^+$ that satisfies $\Psi_0$.
Fix an index $j(0)$ so that $x_{j(0)}$  is assigned $\FALSE$ by~$\sigma$. 
By symmetry of $\Psi_0$, we could assume $j(0)\in \Xi_0$. Let $\ell(0)=0$ and $G(0)=\{x_{j(0)}\}$.
Consider three sequences $\{j(t)\}_{t\ge 0}$, $\{\ell(t)\}_{t\ge 0}$ and $\{G(t)\}_{t\ge 0}$ defined recursively as follows. 
For every positive integer~$t$,
applying the claim to $a=j(t)\in \Xi_{\ell(t)k/2}$, 
in the first case we let 
\[\ell(t+1)=\ell(t)+1, 
\quad j(t+1)=b\in \Xi_{(\ell(t)+1)k/2} 
\quad \text{and} \quad G(t+1)=G(t)\cup \{x_b\};\] 
in the latter case 
we let 
\[
    \ell(t+1) = \ell(t)+2, \quad j(t+1)=c\in \Xi_{(\ell(t)+2)k/2} \quad  \text{and} \quad G(t+1)=G(t)\cup \{x_c, x_d\}.
\]
By induction on~$t$ (with base case $t=0$) 
we conclude that 
\begin{enumerate}
\item $j(t)\in \Xi_{\ell(t)k/2}$ is assigned to $\FALSE$,
\item $\sigma$ assigns all variables in $G(t)$ to $\FALSE$, and
\item $\ell(t) +2 \ge \ell(t+1) \ge \ell(t) +1$ for $t<T$, where 
\[T:=\min\{t>0: \ell(t)=2n/k-1 \text{ or } 2n/k\}\] is a stopping time of $\{j(t), \ell(t), G(t)\}_{t\ge 0}$.
\end{enumerate}
By construction, for all 
$l\in \{0,1,\ldots,\frac{2n}{k}-1\}$, $G(T)\cap (C_{lk/2} \cup C_{(l+1)k/2})$ is not empty. Hence we have proved the first part of the lemma. 

Next we will show that $|G(T)|\ge \frac{2n}{k}$. For this, observe that for $t<T$, 
\[|G(t+1)|=|G(t)|+\ell(t+1)-\ell(t).\] 
Since $|G(0)|=1$, $\ell(0)=0$ and $\ell(T)\ge \frac{2n}{k}-1$, it holds that $|G(T)|\ge 2n/k$.
\end{proof}

\begin{proof} [Proof of
Lemma~\ref{lem: gadget}] 
In the case of $\Psi_0$, first note that for all $i\in N$, $\sigma^+$ satisfies all instances of $W_{i,\ell} \vee \Pi_{i+k/2}$, so $\sigma^+$ satisfies $\Psi_0$. Also, if $\sigma\neq \sigma^+$,
then the lemma follows immediately from  Lemma~\ref{lem:CoverGadget}.  
The case of $\Psi_1$ holds analogously.
\end{proof}

The names of the indices of the variables of~$\Psi_0$ are not very important, and when we generalise the construction in Lemma~\ref{lem: large gadget} it will be useful to consider an arbitrary permutation of them.
Here is the notation that we will use.  
Let $\pi$ be any permutation of~$N=\{0,\ldots,n-1\}$. We use the notation $\pi(i)$ to denote the element in~$N$ that~$i$ is mapped to by~$\pi$.
We will construct a formula $\Psi_0^{\pi}$.
Taking $\id$ to be the identity permutation on~$N$, the formula $\Psi_0$ that was already defined is~$\Psi_0^{\id}$.

For $i\in  \mathbb{N}$, let
\[\Xi_i^{\pi} = \{ \pi(i+j \pmod n) \mid j \in \{0,\ldots, k/2-1\} \}\] 
and let $C_i^{\pi} = \{x_\ell \mid \ell \in \Xi_i^{\pi}\}$.
For $\ell \in \Xi^{\pi}_i$, let $W_{i,\ell}^{\pi}$ be the length-$(k/2)$ clause with variable set~$C_i^{\pi}$ in which $x_\ell$ appears positively and all other variables are negated. 
Let $\Pi_i^{\pi}$ be the length-$(k/2)$ clause with variable set~$C_i^{\pi}$ in which all variables are negated. Then
\[\Psi_0^\pi := \wedge_{i\in N, \ell \in \Xi_i} (W_{i,\ell}^\pi \vee \Pi^\pi_{i+k/2}).\]  
The proof that $\Psi_0^{\pi}  \in \Phi_{n,d,k}$ for any $d\geq k^2/2$ is exactly the same as the case $\pi=\id$.
The formula $\Psi_1^\pi$ is obtained from $\Psi_0^\pi$ by negating all of the literals.
 
The following corollary follows immediately from the proof of Lemma~\ref{lem:CoverGadget} (by renaming the indices using~$\pi$)
and the fact that $C_{0}, C_{k/2}, \dots, C_{2n/k-1}$ are disjoint sets.

\begin{corollary}\label{cor:CoverGadget}
Let $k\ge 4$ be an even integer and let $n\geq k$ be a multiple of $k/2$. 
Let $\pi$ be a permutation of $N$. 
If $\sigma \neq \sigma^+$ satisfies $\Psi_0^\pi$
then there exists a subset $\mathcal{M}^\pi(\sigma)\subseteq \{0,1,\dots, 2n/k-1\}$ of size at least $n/k$ such that 
for all $\ell \in \mathcal{M}^\pi(\sigma)$, $\sigma$ assigns at least one variable in~$C^\pi_{\ell k/2}$ to $\FALSE$.
\end{corollary}
\begin{remark}
 While Corollary~\ref{cor:CoverGadget} does not guarantee the uniqueness of $\mathcal{M}^\pi(\sigma)$, in what follows we define $\mathcal{M}^\pi(\sigma)$ to be the lexicographically smallest set among all the smallest sets satisfying the corollary. Since there are finitely many such sets and they can be lexicographically ordered, $\mathcal{M}^\pi(\sigma)$ is a unique and well-defined set for given $\pi$ and $\sigma$.
\end{remark}

Lemma~\ref{lem: gadget} provides formula $\Psi_0$   with a GAP property applying to the number of $\FALSE$s in  its satisfying assignments. 
In the next lemma, we use  $\Psi_0$  to build a larger formula that amplifies the GAP property to make the gap arbitrarily large.
\begin{lemma} 
    \label{lem: large gadget}
Let $k\ge 4$ be an even integer, let $b>1$ be a real number, and let 
$d$
be an integer satisfying
 \[d\ge k^3\big(1-\frac{1}{b}\big)^{-k/2}.\] 
Let $n\geq k$ be a sufficiently large multiple of~$k/2$.
Then there is a formula $\Psi_2\in\Phi_{n,k,d}$ such that if an assignment $\sigma$ satisfies $\Psi_2$ then  either 
    \begin{enumerate}
        \item $\sigma$ has $n$ $\TRUE$s, or
        \item $\sigma$ has at least $n/b$ $\FALSE$s.  
    \end{enumerate}
    Similarly (by symmetry), there is a formula $\Psi_3 \in \Phi_{n,k,d}$ such that, for any satisfying assignment $\sigma$ of $\Psi_3$, either   $\sigma$ has $n$  $\FALSE$s, or $\sigma$ has at least $n/b$  $\TRUE$s.   
\end{lemma}

\begin{proof}
Fix $k$, $b$, $d$ and $n$ as in the statement of the lemma. Let $N = \{0,\ldots,n-1\}$. The variables of $\Psi_2$ 
are $\{ x_i \mid i\in N\}$.
We will use the following notation. For any set $\Gamma$ of permutations of~$N$,
let $\Psi_2^{\Gamma}$ be the formula with variables $\{ x_i \mid i\in N\}$ and clauses $\bigcup_{\pi \in \Gamma} \{  W_{i,\ell}^\pi \vee \Pi_{i+k/2}^\pi \mid {i \in  N}, {\ell \in \Xi_{i}^\pi}  \} $.
Let \[J_*:= \max\left\{ 2, \left\lfloor \frac{2k}{b}\big(1-\frac{1}{b}\big)^{-k/2}\right\rfloor\right\}\]
and let $\Gamma_{J_*}$ be a set of 
$J_*$ 
permutations of~$N$, each chosen independently and uniformly at random.
Let the permutations in~$\Gamma_{J_*}$ be denoted
$\pi_1,\ldots,\pi_{J_*}$. For each positive integer $J\leq J_*$,
let  $\Gamma_J = \{ \pi_1,\ldots,\pi_J\}$.
The formula that we will construct is $\Psi_2 := \Psi_2^{\Gamma_{J_*}}$.
Since the degree of each formula~$\Psi_0^{\pi}$ is at most $k^2/2$, 
the degree of $\Psi_2$
is at most $d_* := J_* k^2/2$.
Note that 
 \begin{align*}
        d_* &= J_*\cdot \frac{k^2}{2} \le \max\left\{ 
          k^2,  \frac{k^3}{b}\big(1-\frac{1}{b}\big)^{-k/2}
        \right\}
        \le 
        k^3(1-\frac{1}{b})^{-k/2}.
    \end{align*}
We will show that, with positive probability over the choice of   $\Gamma_{J_*}$,  
the formula 
$\Psi_2 = \Psi_2^{J_*}$ 
satisfies the requirements in the lemma statement.
Since $\sigma^+$ 
satisfies every formula $\Psi_0^{\pi}$, it suffices to show that,  
with  positive probability, every satisfying assignment $\sigma\neq\sigma^+$ of $\Psi_2^{J_*}$ has at least $n/b$ variables assigned to~$\FALSE$.

For any set $\Gamma$ of permutations of~$N$, let 
$\notalltrue(\Gamma) \subseteq \Omega_n$
be the set of assignments $\sigma \neq \sigma^+$ that satisfy $\Psi_2^{\Gamma}$.
Recall that $m(\sigma)$ is the number of variables that are assigned to $\TRUE$ by~$\sigma$. 
Let $\minfalse(\Gamma) =  \min\{n-m(\sigma) \mid  \sigma\in \notalltrue (\Gamma) \} $, so that $\minfalse(\Gamma)$ is the minimum number of $\FALSE$ variables in any $ \sigma\in \notalltrue (\Gamma)$.
Since 
$\notalltrue(\Gamma_1) \supseteq 
\notalltrue(\Gamma_2 ) \supseteq \cdots \supseteq 
\notalltrue(\Gamma_{J_*} )$,
 we have $\minfalse(\Gamma_1) \leq \cdots \leq \minfalse(\Gamma_{J_*})$.
It suffices to show that, with positive probability over the choice of  $\Gamma_{J_*}$,   
$\minfalse(\Gamma_{J_*}) \ge n/b$.

Using the fact that $n$ is sufficiently large, we will show that for all $t\in [J_*-1]$ and all $\Gamma_t$ such that $\minfalse(\Gamma_t) < n/b$,
\begin{equation}
\label{eq: E increment}
\E_{\pi_{t+1}}\big[\minfalse(\Gamma_{t}\cup\{\pi_{t+1}\})   
\big] \ge \min\left\{ \minfalse(\Gamma_t) + \frac{3n}{4k} \big(1-\frac{1}{b}\big)^{k/2}, \frac{n}{b} \right\}.    
\end{equation}

By Lemma~\ref{lem:CoverGadget}
(using symmetry to establish the statement for~$\pi_1$ rather than for the identity permutation),
$\minfalse(\Gamma_1) \geq 2n/k \geq  
(3n/(4k))
(1-1/b)^{k/2}$.
Thus by~\eqref{eq: E increment}, it follows that 
\begin{align*}
        \E_{\Gamma_{J_*}}[\minfalse(\Gamma_{J_*})] &\ge  \min\left\{J_* \cdot \frac{3n}{4k} \big(1-\frac{1}{b}\big)^{k/2}, \frac{n}{b} \right\}\\
        &= \min\left\{ \max\left\{ 2, \left\lfloor \frac{2k}{b}\big(1-\frac{1}{b}\big)^{-k/2}\right\rfloor\right\} \cdot \frac{3n}{4k} \big(1-\frac{1}{b}\big)^{k/2}, \frac{n}{b} \right\} \\
        &\ge \min\left\{ 
            2\cdot \frac{3n}{4k}\cdot \frac{k}{2b}, 
            \frac{4k}{3b}\big(1-\frac{1}{b}\big)^{-k/2}\cdot \frac{3n}{4k} \big(1-\frac{1}{b}\big)^{k/2}
            ,\frac{n}{b}
        \right\}  \ge \frac{n}{b}.
    \end{align*}
The second to last inequality needs some explanation. Let $x = (2k/b)(1-1/b)^{-k/2}$. If $x\geq 2$ then $\lfloor x \rfloor \geq x-x/3 = 2x/3$, and this is applied in the middle term of the final min. On the other hand, if $x<2$  
so that the maximum is taken at~$2$, then $(1-1/b)^{k/2} > k/b$ and the first term of the minimum is
\[2\cdot \frac{3n}{4k} (1-1/b)^{k/2} >
2 \cdot \frac{3n}{4k} \frac{k}{b} = \frac{3n}{2b}\geq \frac{n}{b}.\]

Since   $\minfalse(\Gamma_{J_*})$ is bounded from above by $n$, 
the conclusion 
$ \E_{\Gamma_{J_*}}[\minfalse(\Gamma_{J_*})] \geq n/b$ implies that, with positive probability,  $\minfalse(\Gamma_{J_*}) \ge n/b$, completing the proof of the lemma.
    
It remains to prove the lower bound in \eqref{eq: E increment}. We start with some notation.
For every assignment~$\sigma \in \Omega_n\setminus \{\sigma^+\}$ 
let $F(\sigma)$ be the set of indices of variables that are assigned~$\FALSE$ by~$\sigma$. 
For any set~$\Gamma$ of permutations of~$N$, 
and any $\sigma \in \notalltrue(\Gamma)$, 
let  $\smallsetfalse(\sigma, \Gamma)$ 
be the smallest (and lexicographically least, amongst the smallest) non-empty subset of $F(\sigma)$ 
such that the assignment $\sigma'$ with 
$F(\sigma') = \smallsetfalse(\sigma,\Gamma)$ satisfies $\Psi_2^{\Gamma}$. 
Clearly, 
$|\smallsetfalse(\sigma,\Gamma)| \geq \minfalse(\Gamma)$.
For any non-empty set $\smallsetfalse \subseteq N$,    any permutation $\pi$ of~$N$, and any set   $ M\subseteq \{0,1,\dots, \frac{2n}{k}-1\}$,
let \[\notalltrue(\pi,\Gamma, \smallsetfalse, M) = \{ \sigma \in \notalltrue(\Gamma\cup \{\pi\}) \mid \smallsetfalse(\sigma,\Gamma) = \smallsetfalse, \mathcal{M}^\pi(\sigma)=M\}.\]
If $\notalltrue(\pi,\Gamma, \smallsetfalse, M) \neq \emptyset$ let $\extminfalse(\pi,\Gamma,\smallsetfalse,M) = \min \{ n- m(\sigma), \sigma \in \notalltrue(\pi,\Gamma,\smallsetfalse,M)\}$.
(Otherwise, we do not define $\notalltrue(\pi,\Gamma,\smallsetfalse, M)$.)

We are now ready to prove the lower bound in \eqref{eq: E increment}.
Given a fixed  $t\in [J_*-1]$ and a fixed $\Gamma_t$ such that $\minfalse(\Gamma_t) < n/b$, consider the distribution  of $\minfalse(\Gamma_{t}\cup \{\pi_{t+1}\})$ (under the random choice of $\pi_{t+1}$).  
We will find it convenient to also 
fix~$\smallsetfalse$ and $M$. We will show the following condition~\eqref{eq:EincrementS1} 
for every $M\subseteq \{0,1,\dots, \frac{2n}{k}-1\}$ with $|M|\ge n/k$
and every 
non-empty set  $\smallsetfalse \subseteq N$
with $|S| \geq \minfalse(\Gamma_t)$
such that $\notalltrue(\pi,\Gamma, \smallsetfalse, M) \neq \emptyset$,
\begin{equation}
\label{eq:EincrementS1}
\E_{\pi_{t+1}}\big[\extminfalse(\pi_{t+1},\Gamma_t,\smallsetfalse, M)  \big] \ge \min\left\{ \minfalse(\Gamma_t) + \frac{3n}{4k} \big(1-\frac{1}{b}\big)^{k/2}, \frac{n}{b} \right\}, 
\end{equation}
proving~\eqref{eq: E increment}.
 
If $|\smallsetfalse|\geq n/b$ then trivially $\extminfalse(\pi_{t+1},\Gamma_t,\smallsetfalse,M) \geq n/b$ so suppose $|\smallsetfalse| < n/b$.
For every non-negative integer $\ell < 2n/k$, let $Y_\ell(\smallsetfalse)$ be 
the indicator for the event that the random permutation $\pi_{t+1}$ 
makes the intersection $\Xi_{\ell k/2}^{\pi_{t+1}} \cap \smallsetfalse$ empty.
Note that the sets 
in $\{ \Xi_{\ell k/2}^{\pi_{t+1}} \mid 0 \le \ell < 2n/k\}$ are disjoint, so by Corollary~\ref{cor:CoverGadget}, any
formula $\sigma \in \notalltrue(\pi_{t+1},\Gamma_t,\smallsetfalse, M)$   has at least
$|\smallsetfalse| + \sum_{\ell \in M} Y_\ell(\smallsetfalse)$ 
variables assigned to~$\FALSE$.
So we need only show that the expectation of 
$|\smallsetfalse| + \sum_{\ell \in M} Y_\ell(\smallsetfalse)$ (under the choice of the random permutation~$\pi_{t+1}$) is at least the right-hand-side of~\eqref{eq:EincrementS1}.
First, for every non-negative integer $\ell < 2n/k$,  note that
\begin{align*}
   \E_{\pi_{t+1}}[Y_\ell(S)]&=
\Pr_{\pi_{t+1}}[Y_\ell(\smallsetfalse)=1 
] = 
\prod_{r=0}^{k/2-1} \frac{n-|\smallsetfalse|-r}{n-r}\\
&\ge \big(1-\frac{|\smallsetfalse|}{n-k}\big)^{k/2} 
>  \big(1-\frac{1}{b(1-k/n)}\big)^{k/2}
\ge \frac{3}{4} \big(1-\frac{1}{b}\big)^{k/2},
\end{align*}
where the last equality holds when $n$ is large compared to $k$.
We conclude that \[\E_{\pi_{t+1}} \sum_{\ell \in M} Y_\ell(\smallsetfalse) \geq  \frac{3n}{4k}\cdot \big(1-\frac{1}{b}\big)^{k/2}.\]
Therefore, we conclude~\eqref{eq:EincrementS1}
from $|\smallsetfalse| \geq \minfalse(\Gamma_t)$. 
\end{proof}

Before proving Theorem~\ref{thm: impossible4}, we provide the following more general result from which Theorem~\ref{thm: impossible4} is a corollary. 

\begin{theorem}\label{thm: impossible3}
    For all $\beta \neq 0$, 
    let $\alpha=\alpha(|\beta|)$ be any real in $(0,1)$ satisfying
    \begin{equation} \label{eq: alpha def}
        |\beta| + \frac{\alpha}{1-\alpha} \ln{\alpha} + \ln{(1-\alpha)} > 0.
    \end{equation}
    Let  $k \geq 4$ be an even integer and let $\beta^* \in \mathbb{R}$ be such that $|\beta^*|\ge |\beta|$. Let $n$ be a multiple of $k/2$ that is large enough. 
    If
    \begin{equation}\label{eq: alpha bound}
        d \ge k^3 \alpha^{-\frac{k}{2}},
    \end{equation}
    then there exists a formula $\Phi\in \Phi_{n,k,d}$ such that 
    it is impossible to estimate $\beta^*$ from a sample $\sigma \sim \Pr_{\Phi, \beta^*}$ with high probability.
\end{theorem}

\begin{proof}[Proof of Theorem~\ref{thm: impossible3}]
    We first assume $\beta^*>0$ and
    let $\Psi_2$ be the formula given by Lemma~\ref{lem: large gadget}.
    Let $\alpha = \alpha(|\beta|) \in (0,1)$ be a constant satisfying 
    \eqref{eq: alpha def}.
    To see the existence of $\alpha$, note that 
    \[
    f(\alpha):=-\frac{\alpha}{1-\alpha} \ln{\alpha} - \ln{(1-\alpha)}
    \]
    is an increasing bijective function from $(0,1)$ to $(0,\infty)$.
    By \eqref{eq: alpha def}, we have $\beta\ge f(\alpha)$.
    We will show that for any $\beta^* \ge \beta$, samples from both  $\Pr_{\Psi_2, \beta^*}$ will be $\sigma^+$ with  probability $1-e^{-C_1 n}$ for some $C_1>0$. Hence, not only does one-shot learning fail with high probability, but even $e^{C_1n/2}$ many independent
  samples provide no additional information with high probability.

    Setting $b=(1-\alpha)^{-1}$, condition \eqref{eq: alpha bound} becomes $d\ge k^3\big(1-\frac{1}{b}\big)^{-k/2}$. Thus,  Lemma~\ref{lem: large gadget} yields 
    \begin{equation} \label{eq:prob of pf}
        \Pr_{\Psi_2, \beta^*} 
        [\sigma = \sigma ^+] 
        = \frac{Z_1}{Z_1 + \sum_{\alpha=0}^{1-1/b} Z_\alpha }, \text{ where }
        Z_\alpha:=\sum_{\sigma\in\Omega_n: m(\sigma)=\alpha n} e^{\beta^* \alpha n} \mathbbm{1}[\sigma\models \Psi_2].
    \end{equation}
    Using Stirling's approximation $
            \binom{n}{\alpha n} \le \frac{(\alpha^\alpha\cdot (1-\alpha)^{1-\alpha})^{-n}}{\sqrt{2\pi n \alpha(1-\alpha)}}$, we obtain that 
    \begin{equation}
        \label{eq: apply stiling}   Z_\alpha \le e^{\beta^* \alpha n} \binom{n}{\alpha n} \le 
    \exp{\left[ n\cdot 
        \big( 
            \beta^* \alpha - \ln(\alpha^\alpha\cdot (1-\alpha)^{1-\alpha})
        \big)
    \right]}.
    \end{equation}
    From the definition of $\alpha$ in \eqref{eq: alpha def}, it follows that 
    \begin{equation}
        \label{eq: apply alpha def}    \exp{\left[ n\cdot 
        \big( 
            \beta^* \alpha - \ln(\alpha^\alpha\cdot (1-\alpha)^{1-\alpha})
        \big)
    \right]} 
    < e^{\beta^* n}\cdot e^{-\Omega(n)} = Z_1\cdot  e^{-\Omega(n)}.
    \end{equation}
    Combining \eqref{eq:prob of pf} with the estimates \eqref{eq: apply stiling} and \eqref{eq: apply alpha def}, we have 
    \[
    \Pr_{\Psi_2, \beta^*} 
        [\sigma = \sigma ^+] 
        \ge \frac{Z_1}{Z_1 + n\cdot Z_1 \cdot e^{-\Omega(n)}}
        = \frac{1}{1 + n \cdot e^{-\Omega(n)}} 
        = 1 - e^{-\Omega(n)}.
    \]
    Using the formula $\Psi_3$, the proof of the case $\beta^*<0$ is completely analogous. 
\end{proof}
  
Finally, we prove Theorem~\ref{thm: impossible4} as a special case of Theorem~\ref{thm: impossible3}.
\begin{proof}
     We give an explicit $\alpha:(1,\infty) \rightarrow (0,1)$ such that $\alpha(|\beta^*|)$  satisfies \eqref{eq: alpha def} for any $|\beta^*|>1$:
    \begin{equation}
        \label{eq: alpha def 2}
        \alpha(\beta) = 1 - \frac{1}{e^{\beta-1} } = \frac{e^\beta - e}{e^\beta}.
    \end{equation}
    Thus, we obtain the explicit lower bound \eqref{eq: e bound} of $d$ 
    by plugging \eqref{eq: alpha def 2} to \eqref{eq: alpha bound}.

    Next we verify that the choice of $\alpha$ in \eqref{eq: alpha def 2} satisfies \eqref{eq: alpha def}.
    As in the proof of Theorem~\ref{thm: impossible3}, 
    we define 
    $ f(\alpha):=-\frac{\alpha}{1-\alpha} \ln{\alpha} - \ln{(1-\alpha)}$.
    To verify \eqref{eq: alpha def}, it suffices to show that $f(\alpha (\beta )) < \beta$ for all $\beta >1$.
    After some algebraic simplifications, we arrive at 
    \begin{align*}
        f(\alpha(\beta)) &= -\left[ 
           \frac{e^\beta - e}{e^\beta} \cdot \frac{1}{e^{1-\beta} }\cdot \ln\big( \frac{e^\beta - e}{e^\beta} \big)
           + \ln\frac{1}{e^{\beta-1}} 
        \right] \\
        &= -\left[ 
            (e^{\beta - 1} - 1) \cdot
            \big( \ln( e^\beta - e ) - \beta\big)
            + 1 -\beta
        \right]\\
        &= \ln( e^\beta - e ) -1 + \beta e^{\beta -1} -e^{\beta -1}\ln( e^\beta - e )\\
        &= \beta + \ln(1-e^{1-\beta})-1 + \beta e^{\beta -1} - e^{\beta -1 }\big[ \beta + \ln(1-e^{1-\beta}) \big] \\
        &= \beta +   \ln(1-e^{1-\beta}) -1 - e^{\beta -1 } \ln(1-e^{1-\beta}) \\
        &= \beta - 1+ (1 - e^{\beta -1}) \ln(1-e^{1-\beta})
    \end{align*}
    Notice for all $0<x<1$, by Taylor's expansion, 
    \[
        (1-\frac{1}{x}) \ln(1-x) = -(1-\frac{1}{x})\cdot \sum_{n=1}^\infty \frac{x^n}{n} = 1 + \sum_{n=1}^\infty \big( \frac{1}{n+1} - \frac{1}{n} \big) x^n 
        = 1 - \sum_{n=1}^\infty  \frac{x^n}{n(n+1)}<1.
    \]
    Hence, by setting $x=e^{1-\beta}$, we have shown $f(\alpha(\beta))<\beta$.
\end{proof}

\bibliography{arxiv_sub3}

\end{document}